\documentclass[conference]{IEEEtran}
\usepackage{cite}
\usepackage{amsmath,amssymb,amsfonts,amsthm}
\usepackage{dsfont}
\usepackage{algorithm,algorithmic}
\usepackage{graphicx}
\usepackage{textcomp}
\usepackage{xcolor}
\usepackage{subcaption}   
\usepackage{makecell}

\begin{document}
\bstctlcite{IEEEexample:BSTcontrol}

\newcommand{\blue}[1]{\textcolor{teal}{#1}}
\newcommand{\green}[1]{\textcolor{violet}{#1}}
\newcommand{\purple}[1]{\textcolor{purple}{#1}}
\newcommand{\red}[1]{\textcolor{red}{#1}}

\newtheorem{theorem}{Theorem}
\newtheorem{lemma}[theorem]{Lemma}
\newtheorem{definition}{Definition}

\title{
  HyRES: A Hybrid Replication and Erasure Coding Approach to Data Storage \vspace{-0.6cm}
}
\author{
  \IEEEauthorblockN{Daniel E. Lucani$^{1}$, Marcell Feh\'{e}r$^{2}$}
  \IEEEauthorblockA{
    $^{1}$ Department of Electrical and Computer Engineering, Aarhus University, Denmark\\
    $^{2}$ Chocolate Cloud ApS, Denmark\\
    daniel.lucani@ece.au.dk, marcell@chocolate-cloud.cc 
  }
  \vspace{-0.8cm}
}

\maketitle

\begin{abstract}
Reliability in distributed storage systems has typically focused on the design and deployment of data replication or erasure coding techniques.
Although some scenarios have considered the use of replication for \emph{hot} data and erasure coding for \emph{cold} data in the same system, each is designed in isolation.
We propose HyRES, a hybrid scheme incorporates the best characteristics of each scheme, thus, resulting in additional design flexibility and better potential performance for the system.
We show that HyRES generalizes previously proposed hybrid schemes.
We characterize the theoretical performance of HyRES as well as that of replication and erasure coding considering the effects of the size of the storage networks.
We validate our theoretical results using simulations. These results show that HyRES can yield simultaneously lower storage costs than replication, lower probabilities of file loss than replication and erasure coding with similar worst case performance, and even lower effective repair traffic than replication when considering the network size.
\end{abstract}


\section{Introduction}
Over the course of the past decades, research focused on distributed storage systems has provided a rich set of solutions to address different system requirements and delivering efficient and reliable operation.
In particular, the state-of-the-art on reliable storage has focused on two main alternatives: data replication and erasure coding. 
The former is usually a good alternative for \emph{hot data}, where the ability to provide fast access (no computation, potential to perform load balancing) and efficient loss repair as a fragment lost can be repaired by creating a new copy from the remaining replicas. 
However, these benefits come at the cost of additional storage. 
The latter is usually a good alternative for \emph{cold data}, where reducing the cost of storage by (potentially) increasing response time and increasing loss repair costs may be acceptable.
Various code designs to reduce repair costs under various system considerations have been introduced, e.g.,~\cite{Dimakis2010,Hou2019,Zhao2019,Patra2022,Patra2025}.
Recently, distributed systems requiring cost-effective security and reliability have incorporated erasure coded solutions, e.g.,~\cite{Kretsis2021,Kretsis2025}.

Many practical systems typically select one or the other due to operational simplicity or match to the particular use case.
Operating both approaches in the same system have also been pursued in practical settings, but usually to handle different data sizes, e.g., erasure codes for large data and replication for small metadata and keys in Cocytus~\cite{cocytus2017}, or optimized for specific workloads in distributed Cloud storage systems, e.g., ~\cite{legostore2022}.
Comparisons of these schemes based on their worst-case reliability (e.g., minimum number of nodes lost that causes loss of data), storage costs, or repair costs based on the occurrence of one (or more) losses have been carried out. 
However, the state-of-the-art generally lacks a fair cost-benefit comparison between these two approaches that accounts for network size, number of losses experienced by the system, and allocation policies. 

\begin{figure}[t]
  \centerline{\includegraphics[width=0.93\columnwidth]{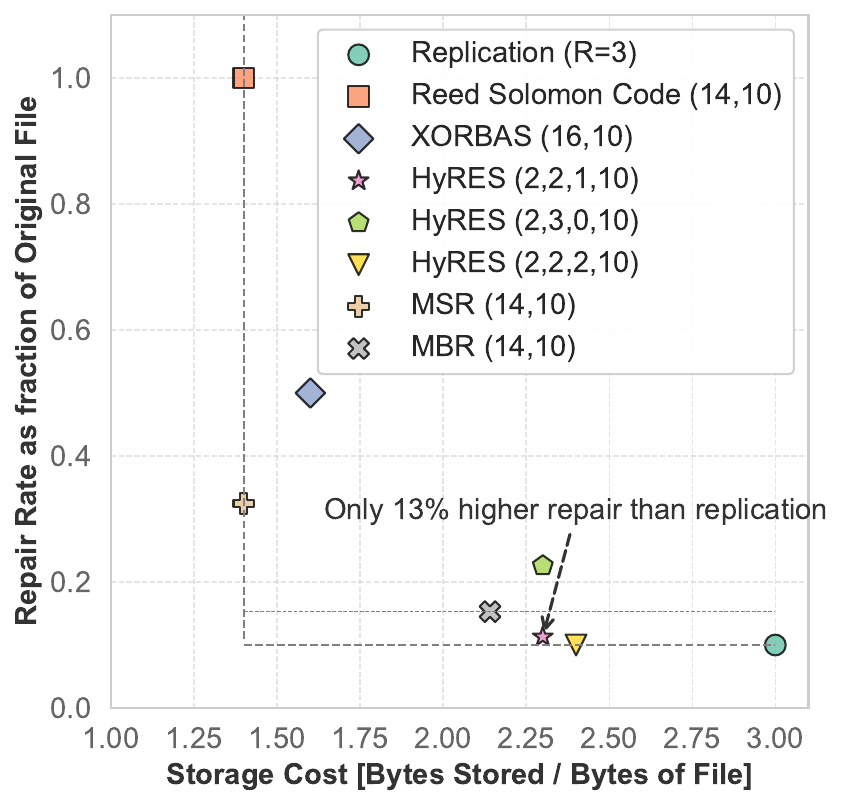}}
  \caption{ Repair Traffic versus Storage Costs for a single node loss of various schemes. Files are split into $k=10$ in the case of replication, XORBAS, Reed-Solomon, and the proposed Hybrid scheme before introducing redundancy. MBR and MSR points of regenerating codes consider the equivalent redundancy of Reed-Solomon code. }
  \label{fig:repair-storage-tradeoff}
\end{figure}

Hybrid approaches, where replicas of fragments and coded fragments are part of the same storage policy have been considered in the literature. 
For example, \cite{Feher2018,Brito2019, hybrid2021} all considered systems that could manage, hot, warm and cold data,
where hybrid approaches were used to manage \emph{warm data}. 
These hybrid schemes considered each fragment of the original file to have a number of replicas while additional erasure coded fragments using random linear network coding~\cite{Ho2006} or Reed-Solomon codes where added to reduce storage costs. 
These schemes also allow for the system to support smoother transitions between the three data states.
The work in~\cite{ZHOU2023} considered its application to two-layer wireless heterogeneous storage networks considering access frequency and repair bandwidth as key focus points.
In contrast to previous proposals, DR-MDS~\cite{DR-MDS2019} provided a scheme that created a Maximum Distance Separable (MDS) code and replicated every piece, including the erasure coded ones. 
This provided an advantage for repair costs, particularly, if few losses occur at a given time.
However, a mathematical characterization and understanding of hybrid schemes, their potential, and fair comparison to erasure coded and replication-based storage systems is still missing.



This paper's goal is two-fold. 
First, we propose a new family of hybrid schemes, called HyRES, to trade-off storage and repair costs by combining replication, erasure coding, and concepts of local repairability inspired in part by~\cite{XORBAS2013}. 
Second, we provide a mathematical framework to compare erasure coded, replication and hybrid schemes considering the effect of network size and occurrence of loss events in the network.
The framework allows for the use of various models for loss probability of the nodes and the special case of i.i.d loss probability per node is further analysed.
Our simulation results and used to validate the theoretical models\footnote{Note that Fig.~\ref{fig:repair-storage-tradeoff} depicts the MBR and MSR schemes for a specific configuration. These are chosen to capture an equivalent worst case node loss protection granted by the hybrid scheme. }.
They show that hybrid schemes can reduce $30$\% or more the overall storage cost of the system while only increasing $13$\% of network use with respect pure data replication when studied in isolation.
Additionally, these schemes are shown to have a much lower loss probability for the entire file as either replication or erasure coding with comparable conditions.

The paper is organized as follows.
We describe the system model in Section~\ref{sec:systemmodel}.
Section \ref{sec:HYRES} introduces HyRES and characterizes its performance, while Section~\ref{sec:analysis} provides mathematical analysis to characterize the behaviour of replication and erasure coding techniques to provide a fair comparison to HyRES.
This analysis incorporates the size of the network in the repair cost of the various systems.
We validate our mathematical analysis with simulations in Section~\ref{sec:evaluation}. 
Section~\ref{sec:conclusions} summarizes our findings.


\section{System Model}\label{sec:systemmodel}

In this work, we will make some key assumptions for our system model. 
We aim to compare different distributed storage schemes fairly by applying equivalent considerations to each one.
This means that our system model will not be limited only to the network or node loss model, but also to aspects of file fragmentation and storage allocation policies for those fragments.
We breakdown these aspects in the following.

\textbf{Network:} We consider a distributed storage system of $N$ nodes. 
Each node can communicate to all others. 
As our focus is on the storage and repair policies, we do not estipulate where the workloads for ingestion of files or redundancy generation are carried out. 

\textbf{File Fragmentation and Redundacy:} Each file is stored by first dividing it into  $k$ equally-sized fragments \footnote{Zero-padding is used in case the the original data size is not divisible by $k$}.
Splitting files into fragments increases read and write speeds due to parallelization and greatly increases flexibility of placement and handling for very large files.
Depending on the storage policy, the fragments may be replicated $R$ times, erasure coded, or a combination of the two.
This generates a total of $n$ fragments per file, where $n>k$ for all cases of interest for reliable storage.

\textbf{Storage Policy:} Each node stores fragments, coded or not, of multiple files.
However, we assume no node stores more than one fragment of the $n$ corresponding to each file.
Thus, $N \geq n$ nodes are required to satisfy this condition.

\textbf{Node Loss:} We assume nodes have the same probability of loss when a loss event occurs (e.g. due to disk failures or individual node failures). 
This model can be used to consider time varying and time-invariant loss probability models. 
We provide analytical results for the latter for simplicity.

\section{HyRES - Hybrid Redundancy Scheme} \label{sec:HYRES}
\begin{figure}[t]
  \centerline{\includegraphics[width=\columnwidth]{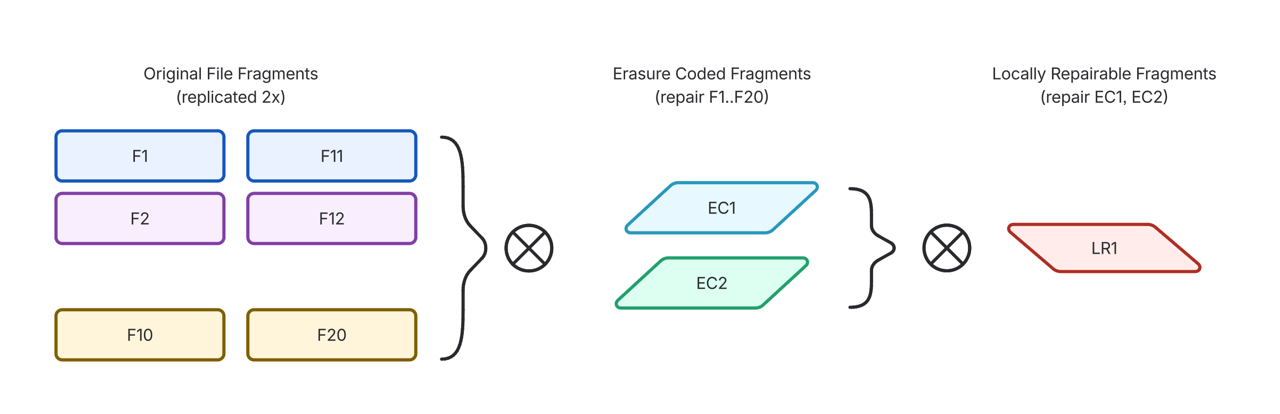}}
  \caption{ Example of HyRES $(2,2,1,10)$ illustrating two replicas of the original fragments, two MDS erasure coded fragments, and a locally repairable linear combination of the two MDS erasure coded fragments}
  \label{fig:HyRES_Example}
\end{figure}
This section proposes HyRES, a hybrid scheme that combines three core ideas on files that are split into equal sized fragments. 
First, replicating the fragments across nodes, which allows fast and zero-compute recovery when any of them are lost.
Second, the use of MDS erasure codes to generate additional coded fragments that can be leveraged if all replicas of one or several fragments are lost.
Third, the use of local repairability for the erasure coded fragments, in order to reduce the amount of traffic generated when losing one of them. 
If two or more coded fragments are lost, this will require the system to revert to decoding missing fragments prior to repairing.
In the following, we define our hybrid scheme and analyze its key characteristics. 
We also discuss how HyRes generalizes and encompasses previous schemes proposed in the state of the art.

 \begin{definition} \label{def:hybrid_scheme}
A HyRES $(R,e,l,k)$ scheme divides each file in $k \geq 2$ fragments to create $R\geq 2$ replicas in total of each fragment. 
It also generates $e\geq 2$ erasure coded fragments generated using an MDS code that takes as input the original $k$ fragments. 
There are also $l$ coded fragments that are linear combinations of $l$ disjoint subsets of the $e$ erasure coded fragments. 
Special considerations are taken for the following corner cases (a) $l=0$ where no additional coded fragments are generated beyond the original $e$; (b) $l=1$, where the newly generated fragment is setup to produce an MDS code with respect to the original $k$ fragments when removing any of the $e$ MDS coded fragments; and (c)
$l=e$ which produces $(R-1)$ direct copies of each of the $e$ MDS erasure coded fragments.
\end{definition}

The total number of fragments generated from the original $k$ fragments is then $R \cdot k + e + l$ fragments.
Our focus on the cases of $l=0,1, e$. Prior work in~\cite{Feher2018,Brito2019, hybrid2021} corresponds to $l=0$, while $l=e$ represents a generalization of the DR-MDS scheme. 
Finally, the $l=1$ case provides a simple and efficient approach for local repairs while keeping good recovery performance.
Figure~\ref{fig:HyRES_Example} provides an example of HyRES $(2,2,1,10)$.

\textbf{Worst Case Loss Performance:} Lemmas~\ref{thm:hybrid_worst_case_loss} and~\ref{thm:hybrid_worst_case_loss_l_e} discuss the worst case loss scenario of the scheme for $l=0,1$ and $l=e$, respectively.

\begin{lemma} \label{thm:hybrid_worst_case_loss}
The worst case number of lost fragments to cause file loss in a HyRES $(R,e,1,k)$ and HyRES $(R,e+1,0,k)$ schemes is $R + e + 1$, i.e., no loss occurs if fewer than $R+e$ fragments are lost.
This worst case loss occurs when $R$ copies of a single fragment are lost and all $e+1$ coded fragments are also lost.
\end{lemma}
\begin{proof} [Proof Sketch] 
It can be verified that the worst case scenario described results in the loss of the file as $R \geq 2$ and $e\geq 2$ by Definition~\ref{def:hybrid_scheme}.
For any other loss pattern having $R+e+1$ losses, this requires that one or more of the coded fragments is not lost.
If $L \in {1,2,...e}$ coded fragments remain, this requires that $R \cdot L$ or more replicated fragments are lost in order to generate an equivalent loss of $L$ erasure coded fragments.
This requires a larger number of losses by Definition~\ref{def:hybrid_scheme}.
For the case of $l=1$, if all erasure coded fragments are not loss, given that one of them is linearly dependent of the other $e$, the system requires at least $R \cdot e$ fragment losses to cause an equivalent loss of the $e+1$ coded fragments.  
But $e + 1 \leq R \cdot e$ for the conditions given in Definition~\ref{def:hybrid_scheme}.
For the case of $l=0$, the system would require at least $R \cdot (e+1)$ fragment losses for the equivalent loss of the same $e+1$.
Using the same argument as for $l=1$, we see that the described loss pattern requires the fewest losses to cause a loss of the file.
\end{proof}

\begin{lemma} \label{thm:hybrid_worst_case_loss_l_e}
The worst case number of lost fragments to cause file loss in a HyRES $(R,e,e,k)$ schemes is $R \cdot (e+1)$.
This worst case loss occurs when $R$ copies of any combination of $e+1$ coded or original fragments are lost.
\end{lemma}
\begin{proof} [Proof Sketch] 
  As coded fragments as also copied, there is no inherent difference/advantage between losing $R$ copies of a coded or uncoded fragment. 
  In contrast to the proof of Lemma~\ref{thm:hybrid_worst_case_loss}, there is no benefit from keeping additional coded fragments beyond the original $e$. 
  Thus, any combination that results in a complete loss of $e+1$ fragments coded or original results in a loss.  
\end{proof}

\textit{Remark:} The DR-MDS~\cite{DR-MDS2019} corresponds to a HyRES $(2,e,e,k)$ scheme. 
When compared to Hybrid schemes with the same number of overall coded fragments, such as HyRES $(2,2e-1,1,k)$ or HyRES $(2,2e,0,k)$, we see that the worst case loss is $2e +2$ in all cases with the same number of total coded fragments, namely, $2e$.
However, there are far more loss patterns that cause this worst case scenario due to the fact that the replicating the erasure coded fragments generates a less robust.
In fact, there are $k$ loss patterns for HyRES $(2,2e-1,1,k)$ or HyRES $(2,2e,0,k)$ that cause the worst case loss, while there are $\binom{k+e}{e+1}$ patterns for DR-MDS.

\textit{Example:} For the case $k=10$, $e=2$ there will be $220$ loss patterns of the worst case while the other schemes would have $10$ loss patterns creating the worst case.

\textbf{Repair Costs:} In the event of losses, if the file is recoverable, HyRES will prioritize transmissions that requires fewest fragments to be accessed for repairing each fragment.
We define these repair costs in terms of the fraction of the size of the original file that needs to be transmitted by the system.
For a given number of lost fragments, we will consider the mean repair costs under the node loss assumption of Section~\ref{sec:systemmodel}.

\begin{lemma} \label{thm:hybrid_repair_cost_one_loss}
The mean repair cost of HyRES for a single fragment loss as a fraction of the original file is (a) $1/k$ for HyRES $(R,e,e,k)$, 
(b) $\frac{R + (e+1) \cdot e/k}{ R \cdot k + e + 1}$ for HyRES $(R,e,1,k)$, and 
(c) $\frac{R + e + 1}{ R \cdot k + e+1}$ for HyRES $(R,e+1,0,k)$.
\end{lemma}

\textbf{File Loss Probability:} In the event of $L$ losses in a network of $N$ nodes, an individual file may be lost with a given probability depending on whether the nodes lost correspond to nodes storing a fragment the HyRES protected file and 
on whether the overall loss experienced by those fragments renders the file unrecoverable. 
Theorems~\ref{thm:hyres_loss_prob_l_1},~\ref{thm:hyres_loss_prob_l_0} and~\ref{thm:hyres_loss_prob_l_e} provide the probability of file loss conditioned on $L$ losses on the network for the key cases of interest in HyRES. 
This conditional probability can be used with various distributions for loss probability with the assumption that the placement of fragments and node losses are independent. 
In the following, $ \mathds{1}_{A}$ represents the indicator function, which is one for condition $A$.

  \begin{theorem} \label{thm:hyres_loss_prob_l_1}
    For the network described in Section~\ref{sec:systemmodel}, the probability of losing a file using a HyRES $(R,e-1,1,k)$, an event called $F_{l=1}$, when $L$ nodes are lost is
    \begin{eqnarray}
      &&{\cal P}_{HyRES,l=1} (R,e, k | L) = \mathds{P}(F_{l=1} | {\cal L} = L; R, e, k) \nonumber \\
      &&= 1 - \frac{\sum_{ (x_1,..,x_k, y_e, y_l,z) \in S_{F_{l=1}} } \binom{N}{x_1, ..., x_k, y_e, y_l,z } }{  \binom{N}{L}} 
    \end{eqnarray}
    where $ x_i \in \{0, ..,R\} \forall i \in \{1, ...,k\}$, $y_e \in \{0, .., e-1\}$, $y_l \in \{0,1\}$, $z \in \{0, ..., N - Rk - e\}$, and 
    $S_{F_{l=1}} = C_0 \cup C_1 $ identifies all patterns of $L$ node losses that do not result in file loss where
    $C_0 = \{ (x_1,...,x_k,y_e, y_l, z) | \sum_{i=1}^k \mathds{1}_{\{ x_i = 0 \}} \leq y_e + y_l < e, \sum_{i=1}^{k} x_i + y_e + y_l + z = N - L \} $
    and
    $C_1 = \{ (x_1,...,x_k,y_e, y_l, z) | \sum_{i=1}^k \mathds{1}_{\{ x_i = 0 \}} < y_e + y_l = e, \sum_{i=1}^{k} x_i + y_e + y_l + z= N - L \} $.
  \end{theorem}
The proof follows a combinatorial argument, where the $x_i$ factors represent the number of remaining replicas of the $i$-th original fragment, $y_e$ represents the remaining MDS erasure coded fragments, $y_l = 1$ indicates the presence of the additional coded fragment for local repair, and $z$ indicates the number of remaining nodes in the network no storing the file. 
Using similar considerations, the Theorems~\ref{thm:hyres_loss_prob_l_0} and~\ref{thm:hyres_loss_prob_l_e} describe the file loss probability for $l=0$ and $l=e$.

\begin{theorem} \label{thm:hyres_loss_prob_l_0}
    For the network described in Section~\ref{sec:systemmodel}, the probability of the event $F_{l=0}$ of losing a file using a HyRES $(R,e,0,k)$ given $L$ node losses is
    \begin{eqnarray}
      &&{\cal P}_{HyRES,l=0} (R,e, k | L) = \mathds{P}(F_{l=0} | {\cal L} = L; R, e, k) \nonumber \\
      &&= 1 - \frac{\sum_{ (x_1,..,x_k, y_e,z) \in S_{F_{l=0}} } \binom{N}{x_1, ..., x_k, y_e,z } }{  \binom{N}{L}}  
    \end{eqnarray}
    where $x_i \in \{0, ..,R\} \forall i \in \{1, ...,k\}$, $y_e \in \{0, .., e\}$, $z \in \{0, ..., N - Rk - e\}$, 
    and $S_{F_{l=0}} = \{ (x_1,...,x_k,y_e,z) | \sum_{i=1}^k \mathds{1}_{\{ x_i = 0 \}} \leq  y_e, \sum_{i=1}^{k} x_i + y_e + z = N - L \} $.
  \end{theorem}

 \begin{theorem} \label{thm:hyres_loss_prob_l_e}
    For the network described in Section~\ref{sec:systemmodel}, the probability of the event $F_{l=e}$ of losing a file using a HyRES $(R,e,e,k)$ given $L$ node losses is
    \begin{eqnarray}
      &&{\cal P}_{HyRES,l=e} (R,e, k | L) = \mathds{P}(F_{l=e} | {\cal L} = L; R, e, k) \nonumber \\
      &&= 1 - \frac{\sum_{ (x_1,..,x_k, y_1, ..., y_e,z) \in S_{F_{l=e}} } \binom{N}{x_1, ..., x_k, y_1, ..., y_e,z } }{ \binom{N}{L}} 
    \end{eqnarray}
    where $x_i \in \{0, ..,R\} \forall i \in \{1, ...,k\}$, $y_j \in \{0, .., R\}  \forall j = \{1, ...,e\}$, $z \in \{0, ..., N - Rk - Re\}$,
    and $S_{F_{l=e}} = \{ (x_1,...,x_k,y_1,...,y_e,z) | \sum_{i=1}^k \mathds{1}_{\{ x_i = 0 \}} \leq  \sum_{j=1}^e \mathds{1}_{\{ y_j = 0 \}} , \sum_{i=1}^{k} x_i + \sum_{j=1}^{k} y_j + z = N - L \} $
  \end{theorem}

\section{Performance of Replication and MDS Codes} \label{sec:analysis}
This section analyses the conditional file loss performance for $R$-way replication and MDS codes. 
As in Section~\ref{sec:HYRES}, exact replication results rely on combinatorial arguments, although a simple upper bound is also developed.
For MDS codes, there exists a simple, closed-form solution for the problem.

\subsection{Replication}
In the case of replication, we consider that each fragment is replicated a total of $R$ times.
Let us calculate the probability of irreperably losing a file, i.e., all replicas are lost for one or more fragments. 
Theorem~\ref{thm:replication} provides an exact result, based on a combinatorial argument as in Section~\ref{sec:HYRES}.
 \begin{theorem} \label{thm:replication}
    For the network described in Section~\ref{sec:systemmodel}, the probability of losing a file using $R$ replicas per fragment is
    \begin{eqnarray}
      &&{\cal P}_r (R | L)  = \mathds{P}(F_R | {\cal L} = L; R) \nonumber \\
      &&= 1 - \frac{\sum_{ (x_1,..,x_k,z) \in S_{F} } \binom{N}{x_1, ..., x_k,z } }{ \binom{N}{L}} 
    \end{eqnarray}
    where $x_i \in \{0, ..,R\} \text{ for } i = \{1, ...,k\}, z \in \{0, .., N-Rk\}$, and
     $S_{F} = \{ (x_1,...,x_k,z) | \sum_{i=1}^k \mathds{1}_{\{ x_i = 0 \}} = 0 , \sum_{i=1}^{k} x_i + z = N - L \} $
  \end{theorem}

For simplicity, we will calculate a simple upperbound. We consider the probability of losing a fragment given $L$ lost nodes in Lemma~\ref{thm:replication_loss_fragment} and use this result to compute an upper bound on the file loss probability in Theorem~\ref{thm:replication_bound}.
\begin{lemma} \label{thm:replication_loss_fragment}
  For a network with $N$ nodes and equal loss probability among nodes, then the probability of the event $A$ of a fragment being irreperably lost when using $R$ replicas, conditioned a number of nodes lost ${\cal L} = L$, is 
  \begin{equation}
    \mathds{P}\left(A | {\cal L} = L; R\right)= \frac{\binom{L}{R}}{\binom{N}{R}}
  \end{equation}
  for $L \geq R$ and zero otherwise.
\end{lemma}
\begin{proof} [Proof Sketch]
Using the concepts of copysets from~\cite{Copysets2013}, we consider all possible copysets for choosing $R$ replicas from $N$ nodes.
The number of copysets with irreversible failures means choosing $R$ elements from the $L$ lost elements. 
\end{proof}

\begin{theorem} \label{thm:replication_bound}
  For the network described in Section~\ref{sec:systemmodel}, the probability of the event $F_R$ of losing a file using $R$ replicas per fragment is upperbounded by 
  \begin{equation}
    {\cal P}_r (R | L)  = \mathds{P}(F_R | {\cal L} = L; R) \leq k \frac{\binom{L}{R}}{\binom{N}{R}}
    \end{equation}
    for $L \geq R$ and zero otherwise.
\end{theorem}
\begin{proof}
  Considering event of losing fragment $i$ is $A_i$, then the probability of losing a file is $\mathds{P}(F_R | {\cal L} = L; R) = \mathds{P}\left( \cup_{i=1}^k A_i | {\cal L} = L; R \right) \leq \sum_{i=1}^k \mathds{P} \left( A_i | {\cal L} = L; R \right) = k \mathds{P}\left( A | {\cal L} = L ; R \right)$ as the probability of losing any fragment is the same. Using Lemma~\ref{thm:replication_loss_fragment} completes the proof.
\end{proof}

\subsection{MDS Codes}
In contrast to the replication case, we can provide a closed-form, combinatorial solution to the problem of a lost file.
We consider a $(n,k)$ MDS code, e.g., a Reed Solomon code, which has $r = n-k$ additional coded fragments. 

\begin{theorem} \label{thm:erasure_coding_loss}
  For the network described in Section~\ref{sec:systemmodel}, the probability of the event $F_{ec}$ of losing a file using a $(n,k)$ MDS code is  
  \begin{eqnarray}
    {\cal P}_{ec} (n, k | L)  &=& \mathds{P}(F_{ec} | {\cal L} = L; n, k)  \\
    &=& \frac{1}{\binom{N}{L}}\sum_{l= n-k+1}^{L} \binom{n}{l} \binom{N-n}{L-l} 
    \end{eqnarray}
    for $L \geq n-k+1$ and zero otherwise.
\end{theorem}
\begin{proof}
  The total number of \emph{deletion sets} is $\binom{N}{L}$ which may cause deletion of a file only if $L \geq n - k +1$ as a $(n,k)$ MDS code can tolerate up to $n - k$ deletions. 
  As the losses may involve nodes with and without fragments for the given file, counting the losses requires to consider the number of scenarios where $n-k+1$ or more node losses involve nodes containing a fragment and the remaining $L-l$ are placed in those without fragments. 
  For only $l$ out of $L$ losses affecting the file fragments, there are $\binom{n}{l} \binom{N-n}{L-l}$ options.
  The proof concludes by adding all cases of $ L \geq l \geq n - k + 1$ and dividing by the total number of deletion sets.
  \end{proof}

\section{Numerical Results}\label{sec:evaluation}
In this section we analyze various performance metrics of the HyRES scheme and compare it to replication and Reed Solomon code (an MDS code).
We also validate analytical results with simulations. 

\begin{figure}[t]
\begin{subfigure}[t]{\columnwidth} 
        \centering
        \includegraphics[width=0.93\columnwidth]{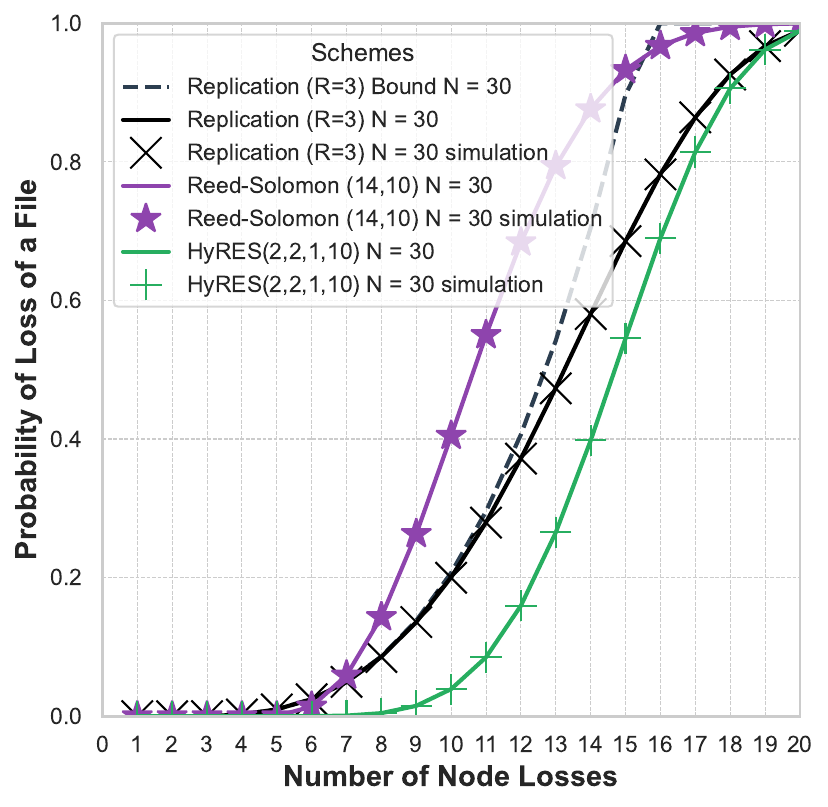} 
        \caption{Network Size $N=30$ nodes}
        \label{fig:sub1}
    \end{subfigure}
  \hfill
\begin{subfigure}[b]{\columnwidth} 
        \centering
        \includegraphics[width=0.93\columnwidth]{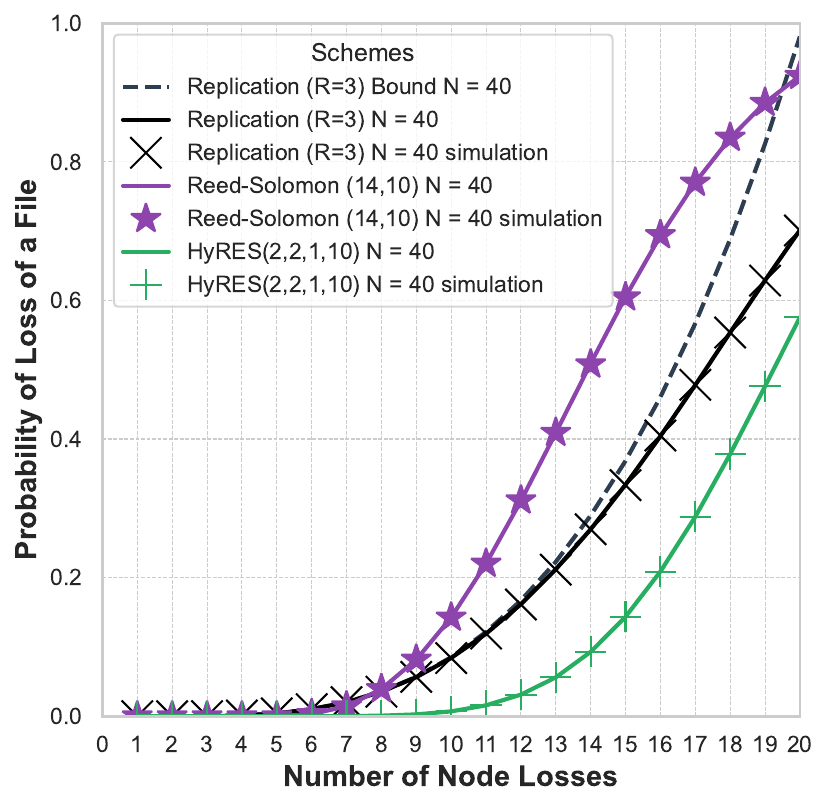} 
        \caption{Network Size $N=40$ nodes}
        \label{fig:sub2}
\end{subfigure}
  \caption{File loss probability for different node losses}
  \label{fig:loss_prob_network_size}
\end{figure}

\textbf{Network Size and Node Loss Effects on Loss Probability:} We study first the file loss probability of the HyRES $(2,2,1,10)$ scheme compared to a $(14,10)$ MDS code and $3$-way replication for different network sizes and node losses in the network in Figure~\ref{fig:loss_prob_network_size}.
Figures~\ref{fig:sub1} and~\ref{fig:sub2} show performance for networks with $N=30$ and $N=40$ nodes, respectively.
These results show that the loss probability of HyRES is much lower compared to classical approaches for all $L$ node losses in each of the networks.
Our results also show that the MDS scheme with the same worst case loss performance is significantly worse than both replication and HyRES.
These simulation results also verify our theoretical calculations and show that the upper bound derived for the replication scheme is tight for small node losses compared.

\textbf{Network Size Effect on Repair Costs for a Single Node Loss:} Although our analysis shows that HyRES schemes have repair traffic that is equal or higher than the repair traffic of a replication scheme with the same $k$ fragments per file, 
the mean repair cost due to a single node loss in the network is a more adequate comparison for all schemes under the same loss event.
Figure~\ref{fig:repair_net_size} shows the average repair traffic for simulations when a single node is lost in the network.
In this case, this traffic is lower the larger the $N$ for all schemes, which is reasonable given the fact that the likelihood of losing a node containing a fragment from the given file is also lowered.
However, Figure~\ref{fig:repair_net_size} also shows that HyRES $(2,2,1,10)$ provides a lower mean repair traffic compared to replication with $R=3$.
This is due to the fact that replication will use a total of $30$ nodes to store the file in this scenario, while HyRES would use only $23$. 
Finally, Figure~\ref{fig:repair_net_size} also shows for comparison the repair traffic measured in the event that a node loss triggers a fragment loss.

\begin{figure}[t]
  \centerline{\includegraphics[width=0.93\columnwidth]{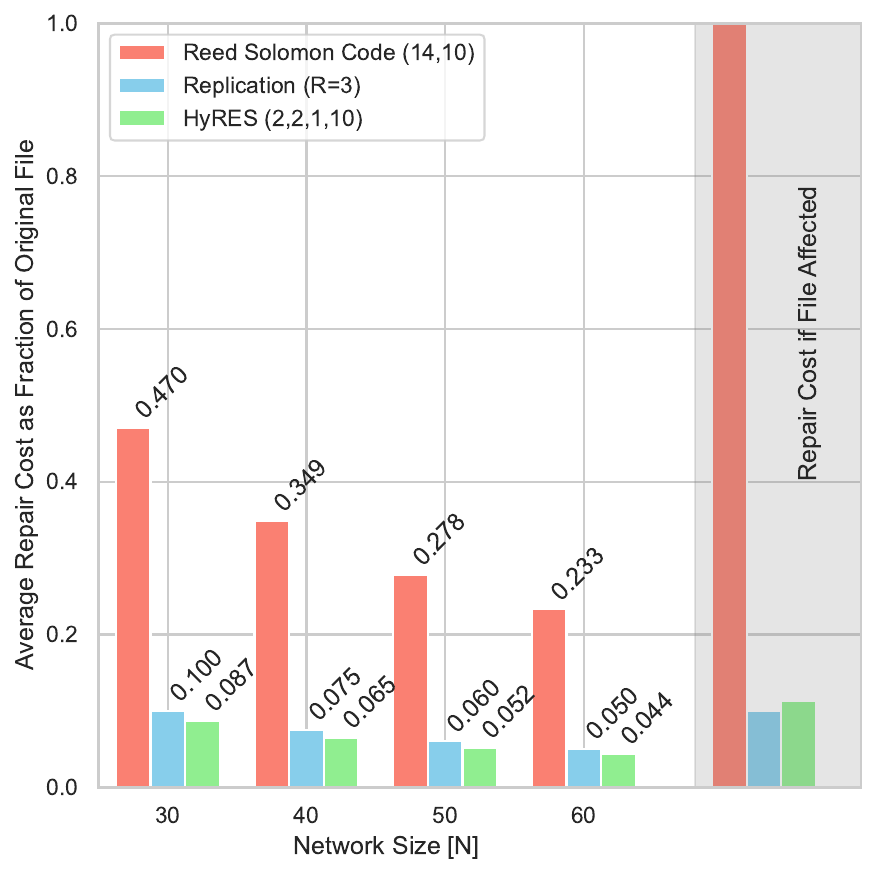}}
  \caption{ Measured repair traffic for a single node loss in the network considering different network sizes}
  \label{fig:repair_net_size}
\end{figure}

\textbf{Node Loss Effects on Repair Traffic:} As simultaneous node losses increase, multiple fragments need to be recovered. 
This could be achieved by repairing from another replica of the same fragment or by collecting $k$ distinct fragments (original or coded) to decode and re-encode the missing fragments in a given node.
If more than one original or coded fragment needs to be recovered, additional single fragment transmissions to the respective nodes will be carried out after decoding.
In the event of a file loss, no repair traffic is triggered. 
Thus, schemes that provide more protection against file losses are likely to experience larger repair traffic under simultaneous losses.
Figure~\ref{fig:rbar_and_line} supports this observation.
For example, the HyRES scheme has a low repair traffic for low node losses, while its repair traffic grows beyond all other schemes for large number of losses.
The MDS erasure codes tested for $L=10$ node losses reduces significantly its repair traffic with respect to $6$ node losses due to the fact that a file has probability higher than $1/2$ to be lost in this scenario.
Figure~\ref{fig:rbar_and_line} shows that the HyRES scheme still incurs in average only a cost of around $0.14$ per lost node at $6$ and $10$ losses, while the replication scheme would always incur a cost $0.1$ per lost node when it can recover the file in this scenario.
Thus, the average repair cost of the scheme per node even for a large number of simultaneous losses is relatively low.
\begin{figure}[t]
  \centerline{\includegraphics[width=0.93\columnwidth]{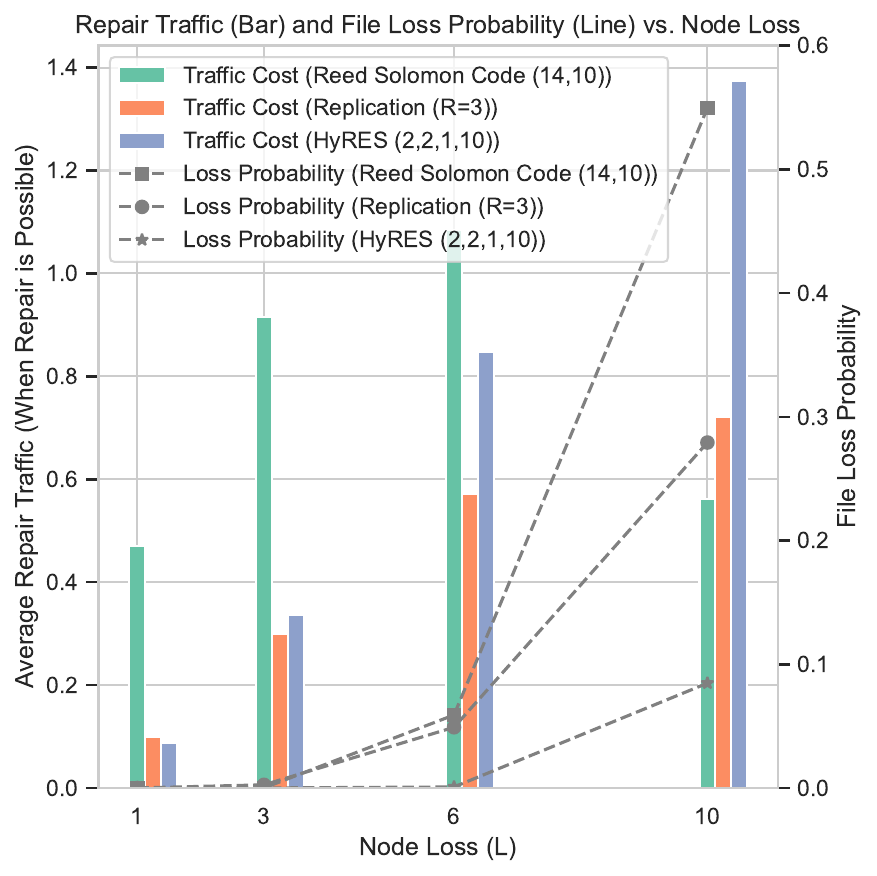}}
  \caption{ Average Repair Traffic per File and File Loss Probability for $N=30$ nodes for various node losses }
  \label{fig:rbar_and_line}
\end{figure}

\section{Conclusions} \label{sec:conclusions}
This work proposed and analyzed HyRES, a family of hybrid schemes for distributed storage to deliver better overall performance in terms of repair traffic, robustness and storage costs. 
We showed that HyRES generalizes previously proposed hybrid schemes and proposes novel alternatives with great practical potential.  
Our work also advocates for studying various schemes considering the network size and number of nodes used by each scheme to achieve a more fair performance analysis.
Using mathematical analysis, performance bounds and simulations, we show that some variants of HyRES can outperform equivalent replication and MDS erasure coded schemes in robustness (file loss probability) and repair traffics, while also providing storage costs considerable improvements with respect to replication.   
Future work may consider heterogeneous characteristics of nodes in the storage system, e.g., different probability of loss for different nodes, or potential structures and hierarcheis in the network, e.g., systems using racks dividing network costs in terms of intra-rack and inter-rack traffic for repair.

\section{Acknowledgement}
This work was financed in part by the Innovation Fund Denmark GreenCOM project under grant number 2079-00040B and by the EU research project EMPYREAN (101136024).
\ifCLASSOPTIONcaptionsoff
  \newpage
\fi
\bibliographystyle{IEEEtran}
\bibliography{ref}

\end{document}